\documentclass[11pt]{article}

\usepackage{paralist}
\usepackage{color}
\usepackage{bm}

\usepackage{soul}
\usepackage[T1]{fontenc}%

\usepackage[cm]{fullpage}
\usepackage{amsfonts}
\usepackage{amssymb}
\usepackage{multirow}
\usepackage{amsmath}
\usepackage{constants}
\usepackage{amsthm}

\usepackage{footnote}
\makesavenoteenv{tabular}
\makesavenoteenv{table}
\makeatletter
\newtheorem*{rep@theorem}{\rep@title}
\newcommand{\newreptheorem}[2]{%
	\newenvironment{rep#1}[1]{%
		\def\rep@title{#2 \ref{##1}}%
		\begin{rep@theorem}}%
		{\end{rep@theorem}}}
\makeatother

\usepackage{ifpdf}
\newcommand{\mydriver}{hypertex}
\ifpdf
\renewcommand{\mydriver}{pdftex}
\fi
\usepackage[breaklinks,\mydriver]{hyperref}

\usepackage{mathtools}
\usepackage[margin=1in]{geometry}

\usepackage{enumitem}
\setlist[enumerate]{itemsep=3pt,topsep=3pt}

\usepackage{thm-restate}

\newcommand{\Thr}{\ensuremath{T}}

\newcommand{\poly}{\textrm{poly}}

\newcommand{\PP}{\mathcal{P}}

\newcommand{\Add}{\textrm{Add}}
\newcommand{\Sum}{\textrm{Sum}}

\newcommand{\nis}{\ensuremath{\mathrm{nis}}}
\newcommand{\nscc}{\ensuremath{\mathrm{nscc}}}
\newcommand{\cnt}{\ensuremath{\mathrm{cnt}}}

\newcommand{\Insert}{\ensuremath{\mathbf{Insert}}}
\newcommand{\Delete}{\ensuremath{\mathbf{Delete}}}
\newcommand{\Query}{\ensuremath{\mathbf{Query}}}

\theoremstyle{plain}
\newtheorem{theorem}{Theorem}[section]%

\newtheorem{lemma}[theorem]{Lemma}
\newtheorem{corollary}[theorem]{Corollary}
\newtheorem{claim}[theorem]{Claim}

\newreptheorem{theorem}{Theorem}
\newreptheorem{lemma}{Lemma}

\theoremstyle{definition}

\newcommand{\ncc}{\ensuremath{\mathrm{ncc}}}

\newcommand{\err}{{\textsc{Err}}}

\newcommand{\junk}[1]{{}}
\newcommand{\pnew}[1]{#1}

\providecommand{\abs}[1]{\lvert#1\rvert}

\usepackage{enumitem}

\title{Constant-Time Dynamic Weight Approximation for Minimum Spanning Forest}

\author{
	Monika Henzinger\footnote{University of Vienna, Faculty of Computer Science, Vienna, Austria. E-mail: \texttt{monika.henzinger@univie.ac.at}. The research leading to these results has received funding from the European
		Research Council under the European Union's Seventh Framework Programme
		(FP/2007-2013) / ERC Grant Agreement no. 340506.}
	\and 
	Pan Peng\footnote{School of Computer Science and Technology, University of Science and Technology of China, China. Email: \texttt{ppeng@ustc.edu.cn}.}
}
\date{}

\begin{document}
	\begin{titlepage}
		\maketitle
		\thispagestyle{empty}

		\begin{abstract}
We give two fully dynamic algorithms that maintain a $(1+\varepsilon)$-approximation of the weight $M$ of a minimum spanning forest (MSF) of an $n$-node graph $G$ with edges weights in $[1,W]$, for any $\varepsilon>0$. 

(1) Our \emph{deterministic} algorithm takes $O({W^2 \log W}/{\varepsilon^3})$ {\em worst-case} update time, which is $O(1)$ if both $W$ and $\varepsilon$ are constants. Note that there is a lower bound by Patrascu and Demaine (SIAM J. Comput. 2006) which shows that it takes $\Omega(\log n)$ time per operation to maintain the {\em exact}
weight of an MSF that holds even in the unweighted case, i.e.~for $W=1$. We further show that any deterministic data structure that dynamically maintains the $(1+\varepsilon)$-approximate weight of an MSF requires super constant time per operation, if $W\geq (\log n)^{\omega_n(1)}$.

(2) Our randomized (Monte-Carlo style) algorithm works with high probability and runs in \emph{worst-case} \pnew{$O(\log  W/ \varepsilon^{4})$} update time if $W= O({(m^*)^{1/6}}/{\log^{2/3} n})$, where $m^*$ is the minimum number of edges in the graph throughout all the updates. It works even against an adaptive adversary. \pnew{This implies a randomized algorithm with worst-case $o(\log n)$ update time, whenever $W=\min\{O((m^*)^{1/6}/\log^{2/3} n), 2^{o({\log n})}\}$ and $\varepsilon$ is constant}. We complement this result by showing that for any constant  $\varepsilon,\alpha>0$ and $W=n^{\alpha}$, any (randomized) data structure that dynamically maintains the weight of an MSF of a graph $G$ with edge weights in $[1,W]$ and $W = \Omega(\varepsilon m^*)$ within a multiplicative factor of $(1+\varepsilon)$ takes $\Omega(\log n)$ time per operation. 
		\end{abstract}
	\end{titlepage}
\section{Introduction}
Minimum spanning forest (MSF) is a fundamental and well studied graph problem in computer science. Given an edge-weighted graph $G$, an {MSF} is a subgraph of $G$ that forms a spanning forest 
and has minimum weight among all spanning forests of $G$, where the weight of a spanning forest is the sum of the edge weights of the forest. %
In this paper, we study dynamic algorithms for maintaining and extracting information regarding the MSF of a dynamically changing graph. %

A (fully) dynamic graph algorithm is a data structure that provides information about a graph property while the graph is being modified by {\em edge updates} such as edge insertions or deletions. When designing a dynamic graph algorithm the goal is to minimize the time per update or query operation. Formally, a \emph{fully dynamic} graph algorithm is an algorithm that maintains some information of a graph $G$ which is undergoing an arbitrary sequence of the following operations: 1) \Insert($u,v, w$): insert the edge $(u,v)$ with weight $w$ in $G$; 2) \Delete($u,v$): delete the edge $(u,v)$ from $G$. (If the considered graph is unweighted, then the weight $w$ of an insertion is always set to be $1$.) The algorithm can further support different \Query~operations depending on the specific graph property. Dynamically maintaining the (minimum) spanning forest (or tree) (e.g., \cite{Fred83,EppsteinGIN97,HenzingerK99,HK97:mst,HLT01:connectivity,HRW15:MST,wulff2017fully,NS17dynamic,NSW17:MST}) and dynamically testing connectivity between any pair of vertices (e.g., \cite{EppsteinGIN97,HLT01:connectivity,HenzingerK99,Thorup00,KapronKM13,GKKT15:connectivity,HuangHKP17}) have played a fundamental role in the area of dynamic graph algorithms. 
The currently best dynamic algorithms for MSF take $O(\log^4 n/ \log \log n)$ {\em amortized}\footnote{An algorithm is said to have amortized update time of $\alpha$ if for any $t$, after $t$ updates the total update time is at most $\alpha t$. It is said to have \emph{worst-case update time} of $\alpha$ if every update time is at most $\alpha$.}  time per update~\cite{HRW15:MST} (that improves upon \cite{HLT01:connectivity}) and $O(n^{o(1)})$ {\em worst-case} time~\cite{NSW17:MST,chuzhoy2019deterministic}.

An approximate version of this problem has also been studied:
A $(1+\varepsilon)$-approximation\footnote{Note that we can easily translate this definition of $(1+\varepsilon)$-approximation of $M$ to the standard one, i.e., a value $M''$ such that $M\leq M''\leq (1+\varepsilon) M$. This can be achieved by multiplying the estimate $M'$ by $1/(1-\varepsilon)$ and replacing $\varepsilon$ by $\varepsilon'=\min\{\varepsilon/4,\frac{1}{2}\}$.} of the weight $M$ of an MSF is a value $M'$ such that $(1-\varepsilon) \cdot M \le M' \le (1 +\varepsilon) \cdot M$. 
In the $(1+\varepsilon)$-approximate MSF problem, the fully dynamic algorithm maintains an MSF of weight $M'$ that is
a $(1+\varepsilon)$ approximation of the weight of an MSF in the current graph. This problem was first studied by Henzinger and King \cite{HenzingerK99}, who also provided a reduction that uses a dynamic algorithm $\cal{B}$ for a spanning tree to obtain a dynamic algorithm $\cal{A}$ for the  $(1+\varepsilon)$-approximate MSF problem, by invoking the algorithm $\cal{B}$ on $\log W/\varepsilon$ subgraphs of a graph $G$ with edge weights in $[1,W]$, for some $W\geq 1$. %
By combining this reduction and best known algorithms for dynamic connectivity (and spanning forest), we have the following table Table \ref{table:known} summarizing the state-of-the-art. We remark that all these algorithms require at least logarithmic update time per operation. 

\begin{table}[h]
	\centering
	\begin{tabular}{c |c p{0.4\textwidth}}
		\hline
		& worst-case & amortized \\
		\hline
		deterministic & $n^{o(1)}$ \cite{chuzhoy2019deterministic}& $\min\{O(\frac{\log W\log^2 n}{\varepsilon\log\log n}), O(\frac{\log^4 n}{\log\log n})\}$ \cite{wulff2013faster,HenzingerK99} \cite{HRW15:MST}\\
		\hline
		Monte Carlo  & $O(\log W\log^4n/\varepsilon)$ \cite{GKKT15:connectivity,HenzingerK99} & \\
		\hline
		Las Vegas& & expected $O(\log W(\log n) (\log\log n)^2/\varepsilon)$ \cite{HuangHKP17,HenzingerK99} \\
		\hline
	\end{tabular}
	\caption{The best known dynamic algorithms for the $(1+\varepsilon)$-approximate MSF problem on a graph with edge weights from $[1,W]$. Those algorithms maintain a spanning tree whose weight is a $(1+\varepsilon)$-approximation of the weight of an MSF. %
		For the Monte Carlo algorithm, there is some $1/\poly(n)$ probability of  answering a query incorrectly. For the Las Vegas algorithm, the answers to the queries are always correct, while the update time is a random variable. } 
	\label{table:known}
\end{table}

\begin{table}[h]
	\centering
	\begin{tabular}{c| c }
		\hline
		& worst-case  \\
		\hline
		deterministic & $O({W^2\cdot (\log W)}\cdot{\varepsilon^{-3}})$ Theorem \ref{thm:deterministic_WMST} \\
		\hline
		Monte Carlo & $O\left(\frac{\log W}{\varepsilon}+\frac{W^3(\log W)\log({W}/{\varepsilon} )\left(\log({\log W}/\varepsilon) +\log n\right)} {\varepsilon^{4}\sqrt{m^*} }\right)$ 
		Theorem \ref{thm:random_WMST}  \\
		\hline
	\end{tabular}
	\caption{Our algorithms for the problem of maintaining a $(1+\varepsilon)$-approximation of the weight of an MSF of a graph with edge weights from $[1,W]$. Our algorithms do \emph{not} maintain a spanning tree. Note that our algorithms are faster than the corresponding algorithms in Table \ref{table:known} for a wide range of parameters. 
	}
	\label{table:ours}
\end{table}

\junk{
	The fastest randomized \emph{amortized} running time
	for this problem
	is $O(\log^2 n (\log \log n)^2/\varepsilon)$ expected amortized update time, which follows from
	\cite{HenzingerK99} and~\cite{HuangHKP17}.
	The fastest randomized \emph{worst-case} running time is
	$O(\log^4n/\varepsilon)$ {\em worst-case} update time, which follows from~\cite{HenzingerK99} and~\cite{GKKT15:connectivity}.  
	The resulting randomized algorithm only works against an oblivious adversary.
}

It is natural to consider a relaxed version of the above problems, by asking how to dynamically maintain an exact or approximate \emph{value}, i.e., not the spanning tree, of the weight of an MSF with much faster update time. %
We first note that maintaining the \emph{exact} value of the weight of an MSF dynamically cannot be done in constant time per operation: There is a lower bound of Patrascu and Demaine~\cite{PatrascuD06}, who
showed that in the cell-probe model many fundamental graph properties, such as asking whether the graph is connected or maintaining the \emph{exact} value of the weight of an MSF, require $\Omega(\log n)$ time per operation, where $n$ is the number of nodes in the graph. Their lower bound even holds for the unweighted case, i.e. for $W=1$.
However, their lower bound does not apply to maintaining an \emph{approximate value} of the weight of an MSF, leading to the following open question: \emph{Can a $(1+\epsilon)$-approximation of the
	weight of an MSF be maintained in time $o(\log n)$?}

We answer this question positively in this paper.
Our main contributions are two dynamic algorithms for maintaining a $(1+\epsilon)$-approximation of the weight of an MSF, that bypass the $\Omega(\log n)$ time lower bound barrier. Even stronger, they have \emph{constant update times} for a wide range of parameters. 
More specifically, our first algorithmic result is as follows. 

\begin{restatable}{theorem}{deterWMSF}~\label{thm:deterministic_WMST}
	There exists a fully dynamic, deterministic algorithm that maintains an estimator $\overline{M}$ that $(1+\varepsilon)$-approximates the weight $M$ of a MSF of a graph with edge weights from $[1,W]$. The worst-case time per update operation is $O({W^2\cdot \log W}\cdot {\varepsilon^{-3}})$. 
\end{restatable}

For constant $W$ and $\varepsilon$ this is a \emph{deterministic worst-case} $O(1)$ time bound. Our algorithm does not require the initial graph to be empty. In contrast, for maintaining the exact value of the weight of an MSF, the best known deterministic algorithm has $O(n^{o(1)})$ worst-case time \cite{chuzhoy2019deterministic} (which actually maintains the MSF, not just the value). 
{In comparison,} the already mentioned lower bound of $\Omega(\log n)$ by~\cite{PatrascuD06} also applies for maintaining the {\em exact} weight of an MSF even for $W=1$. %

We then show that any deterministic data structure that dynamically maintains the $(1+\varepsilon)$-approximate weight of an MSF requires super constant time per operation, if $W\geq (\log n)^{\omega_n(1)}$, where $\omega_n(1)$ is any function that goes to infinity if $n$ goes to infinity. The lower bound on the time per operation for this problem is established in the cell-probe model. %
\begin{restatable}{theorem}{thmdeterlower}\label{thm:deterministiclowerbound}
	Let $W\geq (\log n)^{\omega_n(1)}$. Let $G$ be a dynamic graph with edge weights in $[1,W]$. Any deterministic data structure that dynamically maintains the weight $M$ of an MSF of a graph $G$ within a multiplicative factor of  $(1+\varepsilon)$ for any constant $\varepsilon>0$, or an additive error less than $W/2$ must perform $\omega_n(1)$ cell probes, where each cell has size $O(\log n)$.
\end{restatable}

Our second algorithmic result is a randomized dynamic algorithm that works even against an adaptive adversary\footnote{That is, the adversary sees the answers to all query operations before deciding which edge to update next. We say an adversary is \emph{oblivious}, if it fixes the sequence of edge insertions and deletions and is oblivious to the randomness of algorithm.}. We let $m^*$ be the minimum number of edges of the graph throughout all the updates. In this paper, {\em ``With high probability''} means ``with probability at least $1-\frac{1}{n^c}$ for some constant $c \ge 1$''.
\begin{restatable}{theorem}{randomWMSF}~\label{thm:random_WMST}
	Let $\varepsilon\in(0,1)$ and $W\geq 1$. There exists a fully dynamic algorithm that with high probability, maintains an estimator $\overline{M}$ that is a $(1+\varepsilon)$-approximation of the weight $M$ of MSF of a graph $G$ with edge weights from $[1,W]$. The worst-case time per update operation is \[\pnew{O\left((\log W)/\varepsilon+{W^3(\log W)\log({W}/{\varepsilon} )\left(\log({\log W}/\varepsilon) +\log n\right)}/({\varepsilon^4 \sqrt{m^*}})\right)},\] where $m^*\geq 1$ is the minimum number of edges in $G$ throughout all the updates. Our algorithm works against an adaptive adversary.
\end{restatable}

From the above, we can easily obtain a randomized algorithm that runs in \pnew{sub-logarithmic} time (for constant $\varepsilon>0$) for the setting that $W$ might not be constant. 

\begin{corollary}\label{cor:rand_WMST}
	Let $\varepsilon\in (\frac{1}{n^4},1)$. There exists a fully dynamic algorithm that, with high probability, maintains an estimator $\overline{M}$ that is a $(1+\varepsilon)$-approximation of the weight $M$ of an MSF of a graph $G$ with edge weights from $[1,W]$ such that  \pnew{$W=O((m^*)^{1/6}/\log^{2/3} n)$}, and has 
	\begin{itemize}
		\item \pnew{$O((\log W)/\varepsilon^{4})$} worst-case time per update operation; 
		\item \pnew{$O(1)$} worst-case time per update operation, if  $W$ and $\varepsilon$ are constant;
		\item \pnew{$o(\log n))$} worst-case time per update operation, if  \pnew{$W=2^{o(\log n)}$ and $\varepsilon$ is constant}.
	\end{itemize}
	The algorithm works against an adaptive adversary. 
\end{corollary}
\begin{proof}
	By applying Theorem \ref{thm:random_WMST} with the assumptions that $\varepsilon\in (\frac{1}{n^4},1)$ and that $W=O((m^*)^{1/6}/\log^{2/3} n)$, we obtain a corresponding algorithm with worst-case update time 
	\pnew{	\begin{align*}
			&O\left(\frac{\log W}{\varepsilon}+\frac{W^3(\log W)\log (W/\varepsilon)(\log({\log (W/\varepsilon)})+\log n)}{\varepsilon^4\cdot \sqrt{m^*}}\right)\\
			=&O\left(\frac{\log W}{\varepsilon}+\frac{\sqrt{m^*}(\log W)\log (W/\varepsilon)(\log n)}{\varepsilon^4\cdot \sqrt{m^*}\cdot \log^2 n}\right)
			=O(\frac{\log W}{\varepsilon^4}),
		\end{align*}
		where we used the fact that $m^*\leq O(n^2)$.}
	This finishes the proof of the corollary.
\end{proof}
The requirement that $W$ needs to be small compared to the minimum number of edges in the graph might look artificial. However, we present the following lower bound for any randomized data structure that dynamically maintains the approximate weight of an MSF, which shows that some requirement of this type is necessary. %
\begin{restatable}{theorem}{thmlower}\label{thm:lower}
	For any constant $\varepsilon,\alpha>0$ and $W=n^{\alpha}$, 
	any data structure that dynamically maintains with high probability the weight of an MSF of a graph $G$ with edge weights in $[1,W]$ and  $W=\Omega(\varepsilon m^*)$ within a multiplicative factor of  $(1+\varepsilon)$, or an additive error less than $W/2$, must perform $\Omega(\log n)$ cell probes, where each cell has size $O(\log n)$.
\end{restatable}
This stands in interesting contrast to the above Corollary \ref{cor:rand_WMST}, which \pnew{has update time $o(\log n)$, but requires $W$ to be $\min\{O((m^*)^{1/6}/\log^{2/3} n), 2^{o({\log n})}\}$}.%

We remark that main technical contributions are our algorithmic results given by Theorem \ref{thm:deterministic_WMST} and \ref{thm:random_WMST}, which are summarized in Table \ref{table:ours}, for ease of comparison. %

\subsection{Our techniques}
Both our deterministic and randomized dynamic graph algorithms for this problem use an approach developed in the area of property testing: Build an efficient algorithm for estimating the number of connected components (CCs) in a graph and apply it to suitable subgraphs of $G$~\cite{CRT05:MST,CS09:MST,AGM12:linear}.
More specifically, we build constant-time dynamic algorithms  that estimate the number of
CCs with appropriate additive error, apply them to $O(\log W/\varepsilon)$ many subgraphs,
and then use an extension of the formula in~\cite{CRT05:MST} to disconnected graphs to estimate the weight of an MSF. {We stress that most of previous work uses the formula for \emph{connected} graphs, while in our setting, it is more natural to remove the requirement that the dynamic graph is always connected. Thus, we  extend the formula to work for disconnected graphs (see Lemma \ref{lem:mst_ncc}).
	To use this formula our algorithms need to estimate the number of CCs with an additive error $\varepsilon'\cdot \nis(G)$, where $\nis(G)$ is the number of \emph{non-isolated} vertices in $G$ (see below why this is crucial). Our randomized dynamic algorithm for this problem achieves such an error in time 
	$O(\max\{1,{\log(1/\varepsilon')\log n}/(\varepsilon'^3\cdot \sqrt{m^*})\})$ with high probability, and our deterministic algorithm achieves the same error $\varepsilon' \cdot \nis(G)$ in time $O((1/\varepsilon')^2)$. 
	
	{\bf The randomized algorithm.}
	The \emph{randomized} algorithm 
	for estimating the number of connected components 
	is based on the following {simple approach} (used previously e.g.~\cite{GuptaP13}):
	{\em Whenever (1) there exists a static algorithm that in time $T$ estimates a desired parameter (here the number of CCs) with an additive error of $\err$
		and (2) 
		each update operation changes the value of a desired parameter only by an additive value up to $+/- \delta$ (here 1),
		then  running the static algorithm every $\frac{\err}{\delta}$ update
		operations leads to a dynamic algorithm with additive error of at most $2 \err$ and {\em amortized} time $O(\frac{T \delta}{\err})$ per update.}

	To turn this into a {\em worst-case} time bound observe that we can run the static algorithm ``in the background''  executing $ O(\frac{T \delta}{\err})$ steps of the static algorithm at every update operation. This increases the additive error only by a constant factor.
	We make use of the static (constant-time) algorithm of~\cite{BKM14:numcc} (that improves upon~\cite{CRT05:MST}) for estimating the number of CCs 
	with additive error $\varepsilon' n$ as a subroutine. However, directly invoking this static algorithm we can obtain only a dynamic estimator for the number of CCs with an additive error $\varepsilon' n^{2/3}\log^{2/3}n$ with $O(1/\varepsilon'^{3})$ update time (see Appendix \ref{app:note_ncc}). 
	This is not sufficient for the MSF approximation algorithm: Specifically, we need to be able to dynamically approximate the number of CCs with an additive error $\varepsilon'\cdot \nis(G)$ with fast update time and $\nis(G)$ can be arbitrarily small. %
	To achieve this smaller additive error, we carefully choose different values of $\err$ throughout all the updates and  sample the non-isolated vertices uniformly at random in the dynamic graph. We give more detailed discussions on some other technical difficulties (e.g., we need to  ``synchronize'' the updates in $G$ and its subgraphs to achieve the desired update time) and how we handle them in Section~\ref{subsec:randomized}. 
	
	Sampling non-isolated vertices uniformly at random in a dynamic graph is exactly the problem solved by $\ell_0$-sampling in streaming algorithms (see e.g. \cite{jowhari2011tight}). However, all such algorithms, while only using $O(\poly\log n)$ space, require time $\Omega(\log n)$ per operation. We give a relatively simple data structure (see Lemma \ref{lemma:samplenonzero}) that allows to sample all non-zero entries in a dynamically changing vector of size $n$ in constant time (no matter how small their number might be), albeit with space $O(n)$. We believe that our data
	structure might be of independent interest. 
	
	{Our randomized algorithm uses ``fresh'' random bits whenever it invokes the static (constant-time) algorithm on the current graph, and does not reuse any information computed before. Furthermore, our analysis assumes that the adversary changes the graph in the worst possible way, i.e., changes the number of CCs by at most $1$ in each update. Thus,  our algorithm works against an adaptive adversary, i.e. an adversary that sees the answers to all queries {\em before} deciding on the next update operation.}

	{\bf The deterministic algorithm.}
	To design a {\em deterministic worst-case} dynamic algorithm we cannot simply invoke the static constant-time algorithm: that algorithm is inherently randomized as it is designed with the goal of reading the smallest possible portion of the graph.
	Instead we carefully implement the random local exploration that underlies the static randomized algorithm in a deterministic way. 
	Our key observations are that {\em (1) we only need to count the number of CCs which are small in size, i.e., which consist of up to $1/\varepsilon'$ vertices,
		as the number of larger CCs is at most $\varepsilon' \cdot \nis(G)$ {\em and} (2) these counts can be maintained in  worst-case
		time $O(1/\varepsilon'^2)$ after each update by exploring a neighborhood of $O(1/\varepsilon')$ vertices ``near'' the endpoints of the updated edge.}
	
	Both the randomized and the deterministic MSF algorithm run their respective CCs estimation algorithms on each of the $O({\log W}/{\varepsilon})$ relevant subgraphs with $\varepsilon' = \Theta(\varepsilon/W)$. 
	{Using 
		the above-mentioned formula we obtain an estimate $\overline{M}$ that approximates the weight $M$ of MSF with
		an additive error of $\varepsilon \nis(G)/4$.}
	As the weight of any MSF is at least $\nis(G)/2$, this additive error is at most $\varepsilon M/2$, i.e.,~a
	$(1+\varepsilon)$-approximation of $M$. {Whenever a query on $M$ is asked, we return $\overline{M}$ in constant time.}
	For our deterministic algorithm for MSF, the time per edge update is $O(1/\varepsilon'^2) = O(W^2/\varepsilon^2)$ for each of  the $O({\log W}/{\varepsilon})$ subgraphs,  resulting in a 
	worst-case $O(W^2 \log W/\varepsilon^3)$ update time. The running time of our randomized algorithm for MSF can be analyzed analogously. %

	We remark that 
	{our algorithms are much simpler than} the fastest dynamic (exact or approximate) MSF algorithms: the algorithms of~\cite{HRW15:MST,HLT01:connectivity,GKKT15:connectivity,HuangHKP17,HenzingerK01}
	maintain a hierarchical decomposition with $O(\log n)$ levels (in the approximate setting this is even $O(\log n \log W)$ levels), the algorithm of~\cite{NSW17:MST} maintains a decomposition of the graph into
	expanders and a ``remaining'' part.

	\subparagraph{The lower bounds.} Both the lower bounds given in Theorem \ref{thm:deterministiclowerbound} and Theorem \ref{thm:lower} are proven via simple reductions from the previous cell probe lower bounds for dynamic connectivity. Let us take the proof of Theorem \ref{thm:deterministiclowerbound} for example. We give a reduction, similar to the one in \cite{henzinger1998lower}, from a problem called the \emph{parity prefix sum problem} to the problem of dynamically maintaining the weight of an MSF of a graph with large maximum edge weights within a multiplicative factor of $1+\varepsilon$. That is, using our reduction, one can use an algorithm for the latter problem to solve the former problem, for which a cell-probe lower bound is long known \cite{FS89:cell}. This gives a corresponding lower bound for the problem of dynamically maintaining a $(1+\varepsilon)$-approximation of the weight of an MSF of a graph. The proof of Theorem \ref{thm:lower} is similar, except that we use a variant of reduction from the problem of dynamic connectivity, whose cell-probe lower bound is given in \cite{PatrascuD06}. 

	\subsection{Other related work}\label{sec:related_work}
	There are few other graph problems for which constant-time dynamic algorithms are known: (a) maximal matching (randomized)~\cite{Sol16:matching}, (b) $(2+\varepsilon)$-approximate vertex cover (deterministic)~\cite{BhattacharyaK2019}, (c)  $(2k-1)$-stretch spanner of size $O(n^{1+\frac{1}{k}}\log^2 n)$ for \emph{constant} $k$ (randomized)~\cite{BKS12:spanner}, and (d)
	$(\Delta+1)$-vertex coloring~\cite{henzinger2019constant,bhattacharya2019fully}.  All these are \emph{amortized} time bounds and almost all these algorithms except the coloring algorithm in \cite{henzinger2019constant} maintain a sophisticated hierarchical graph decomposition, which makes them rather impractical. %

	Fully dynamic algorithms for minimum spanning forest and connectivity testing have also been investigated for some special classes of graphs (e.g.,\cite{eppstein1992maintenance,eppstein1996separator}). There are also works on approximating the weight of minimum spanning tree (or forest) in the query model \cite{CRT05:MST,CS09:MST} in sublinear time and data streaming model \cite{AGM12:linear,huang2019dynamic,PS18:stream} in sublinear space. 

	\section{Preliminaries}\label{sec:preliminaries}
	Let $G$ be a (static) undirected graph without
	parallel edges and with edge weights in $[1,W]$. Let $M$ be the weight of an MSF of $G$. We use {\em $\ncc(G)$} to denote the number of CCs of $G$, $\nis(G)$ to denote the number of \emph{non-isolated} vertices of $G$, and {\em size} of a CC  to denote the number of vertices in the CC. 
	
	Our algorithms exploit a relation between the weight of an MSF of a graph $G$ and the number of CCs of some subgraphs of $G$.

	For any $\varepsilon>0$, we let $r=\lceil\log_{1+{\varepsilon}} W\rceil$, $\ell_i=(1+{\varepsilon})^i$, and $\lambda_i=(1+{\varepsilon})^{i+1}-(1+{\varepsilon})^i$. For any $i\geq 0$, we let $G^{(i)}$ denote the subgraph of $G$ spanned by all edges with weights at most $\ell_i$, and let $c^{(i)}$ denote the number of CCs in $G^{(i)}$. %
	We will make use of the following lemma, whose proof is almost the same as in \cite{AGM12:linear} (see also \cite{CRT05:MST,CS09:MST}), except that we are considering a graph that is not necessarily connected. We present the proof here for the sake of completeness. 
	\begin{lemma}[\cite{CRT05:MST,CS09:MST,AGM12:linear}]\label{lem:mst_ncc} 
		Let $\varepsilon\in(0,1)$. Let $G$ be a weighted graph and let $M,W,$ and $c^{(i)}$ be defined as above.  Then  
		\begin{eqnarray}
			M\leq n-c^{(r)} \cdot (1+{\varepsilon})^r + \sum_{i=0}^{r-1} \lambda_i\cdot c^{(i)}\leq (1+{\varepsilon}) M. \label{eqn:app_mst}
		\end{eqnarray}	
	\end{lemma}
	\begin{proof}%
		We let $G'$ be the graph that is obtained by rounding each edge weight up to the nearest power of $(1+\varepsilon)$. Note that the weight $M'$ of an MSF of $G'$ satisfies that $M\leq M'\leq (1+\varepsilon) M$. %
		
		Now we invoke the Kruskal's algorithm on $G'$ to compute $M'$. The algorithm will add $n-c^{(0)}$ edges of weight $\ell_0=1$,  $c^{(0)}-c^{(1)}$ edges of weight $\ell_1$, and so on. Thus
		\[
		M'=n-c^{(0)} + \sum_{i=0}^{r-1} \ell_{i+1} (c^{(i)}-c^{({i+1})})=n-c^{(r)} \ell_r + \sum_{i=0}^{r-1}(\ell_{i+1}-\ell_{i}) \cdot c^{(i)}.
		\] 
		This proves the statement of the lemma.
	\end{proof}
	
	\section{A Deterministic Dynamic Algorithm}\label{sec:parameters} 
	In this section, we present our deterministic dynamic algorithm for maintaining the weight $M$ of a minimum spanning forest of 
	a graph $G$ without
	parallel edges and with edge weights in $[1,W]$, and give the proof of Theorem~\ref{thm:deterministic_WMST}.

	We first give a deterministic dynamic algorithm for estimating the number of connected components (CCs) with appropriate additive error. %

	\begin{restatable}{theorem}{deterNCC}\label{thm:deterministic_NCC}
		Let $\varepsilon>0$. There exists a fully dynamic and deterministic algorithm that preprocesses a potentially non-empty graph in $O(\frac{n}{\varepsilon})$ time, and maintains an estimator $\overline{c}$ s.t., $|\overline{c}-\ncc(G)|\leq \varepsilon \cdot \nis(G)$ with worst-case $O(1/\varepsilon^2)$ update time per operation. 
	\end{restatable}
	
	We remark that in the above theorem, the initial graph can be an arbitrary graph, and the performance guarantee holds even if the algorithm is not aware of the value $\nis(G)$.  The proof of the above theorem is deferred to the end of this section. 
	By combining %
	Theorem~\ref{thm:deterministic_NCC} and the relation in Lemma~\ref{lem:mst_ncc}, we are ready to give the algorithm \textsc{DeterEstWMSF} for maintaining an estimator of the weight of an MST.
	
	\begin{center}
		
		\begin{tabular}{|p{\textwidth}|}
			\hline
			\textsc{DeterEstWMSF}($G,\varepsilon,W$) $\triangleright$ \textbf{Dynamically Maintaining an estimator $\overline{M}$ of the weight of an MSF}%
			\begin{enumerate}
				\item Let $r=\lceil\log_{1+{\varepsilon/2}} W\rceil$, $\ell_i=(1+{\varepsilon/2})^i$, and $\lambda_i=(1+{\varepsilon/2})^{i+1}-(1+{\varepsilon/2})^i$.
				\item Let $G^{(i)}$ denote the subgraph of the current graph $G$ spanned by all edges with weights at most $\ell_i$. Let $c^{(i)}$ be the number of CCs of $G^{(i)}$.
				\item 		
				For each $1 \leq i\leq r$, invoke the dynamic algorithm from Theorem~\ref{thm:deterministic_NCC} to maintain an estimator $\overline{c}_i$ 			
				for $c^{(i)}$ with $\varepsilon' = 
				\varepsilon /(12W)$. %
				\item Define $\overline{M}:=n-\overline{c}_r\cdot (1+{\varepsilon}/{2})^r+\sum_{i=0}^{r-1}\lambda_i\cdot \overline{c}_i$.
			\end{enumerate}
			
			\\ \hline
		\end{tabular}
	\end{center}
	
	Now we are ready to prove Theorem \ref{thm:deterministic_WMST}, which is restated in the following for the sake of readability.
	\deterWMSF*

	\begin{proof}
		Let $r,\ell_i, \lambda_i, G^{(i)}, c^{(i)}, \overline{M}$ be defined as in the above algorithm \textsc{DeterEstWMSF}($G,\varepsilon,W$). Note that $\nis(G)\ge \nis(G^{(i)})$, since $G^{(i)}$ is a subgraph of $G$ for any $1 \le i \le r$.

		Since $\overline{c}_i$ is the estimator of $c^{(i)}$ obtained by invoking the dynamic algorithm in Theorem~\ref{thm:deterministic_NCC} with approximation parameter $\varepsilon' = 
		\varepsilon /(12W)$, we have that  for each $1 \leq i\leq r$, 
		$$|\overline{c}_i-c^{(i)}|\leq\varepsilon\cdot \nis(G^{(i)})/({12W}) \leq {\varepsilon \cdot \nis(G)}/({12 W})$$

		Let $X=n-c^{(r)}\cdot (1+{\varepsilon}/{2})^r+\sum_{i=0}^{r-1}\lambda_i\cdot c^{(i)}$. Recall that $M$ is the weight of an MSF of the current graph $G$. We have 
		\begin{align*}
			|\overline{M}-X| &\leq (1+\varepsilon/2)^r\abs{\overline{c}_r-c^{(r)}}+\sum_{i=0}^{r-1}\lambda_i\abs{\overline{c}_i-c^{(i)}}\\
			&\leq (1+\varepsilon/2)^r \cdot (\varepsilon \cdot \nis(G)/(12W)) +  (\varepsilon/2)\sum_{i=0}^{r-1}(1+\varepsilon/2)^{i}\cdot (\varepsilon \cdot \nis(G)/(12W))\\
			& < (1+\varepsilon/2)^r (\varepsilon \cdot \nis(G)/(12W)) + (1+\varepsilon/2)^r (\varepsilon \cdot \nis(G)/(12W))\\
			&\leq 2 (1+\varepsilon/2)^{1+\log_{1+{\varepsilon/2}}W} \cdot ({\varepsilon \cdot \nis(G)}/(12W))\\
			&\leq 2\cdot (3/2)\cdot ({\varepsilon \cdot \nis(G)}/12)\leq \varepsilon \cdot \nis(G)/4 \leq \varepsilon \cdot M/2,
		\end{align*}
		where the last inequality follows from the fact that $M\geq {\nis(G)}/{2}$, as each non-isolated vertex is incident
		to at least one edge (of weight at least $1$) of any MSF in the simple graph $G$.
		
		By Lemma~\ref{lem:mst_ncc} with approximation parameter $\varepsilon/2$, we have that $M\leq X\leq (1+\varepsilon/2)M$. Thus, $\overline{M}$ is a $(1+\varepsilon)$-approximation of $M$.%

		Now note that for each $1 \leq i\leq r$, for estimating $c^{(i)}$ we invoke the dynamic algorithm in Theorem \ref{thm:deterministic_NCC} with $\varepsilon' = 
		\varepsilon /(12W)$, which has worst-case time $O(1/\varepsilon'^2)=O({W^2}/{\varepsilon^2})$ per update operation. 
		Therefore, the worst-case time per update operation of \textsc{DeterEstWMSF}($G,\varepsilon,W$) is $\sum_{i=1}^r O(\frac{W^2}{\varepsilon^2})=O(\frac{r}{ \varepsilon^2}W^2)=O(\frac{W^2\cdot \log W}{\varepsilon^3})$. %

	\end{proof}

	\subsection{Proof of Theorem \ref{thm:deterministic_NCC}}
	In the following, we present the proof of Theorem \ref{thm:deterministic_NCC}, which we restate below.  
	\deterNCC*
	\begin{proof}
		For a graph $G$, we let $\nscc(G)$ denote the number of CCs of size at most $1/\varepsilon$ in $G$. We first observe that to approximate $\ncc(G)$ of a dynamic graph $G$ with an additive error $\varepsilon \cdot \nis(G)$, it suffice to compute and maintain $\nscc(G)$. %
		This is true as the total number of CCs of size larger than $\frac{1}{\varepsilon}$ is at most $\varepsilon \cdot \nis(G)$, where $\nis(G)$ is the number of non-isolated vertices of $G$. In the following, we first give a static algorithm for computing $\nscc(G_0)$, for any initial graph $G_0$, and then we show how to maintain $\nscc(G)$ dynamically.

		\paragraph{Preprocessing} 	
		We first give a \emph{static} algorithm 	\textsc{StaticNCC} for computing the number of small CCs of any initial graph $G_0$. 
		We maintain a set of $ 1/\varepsilon $ counters  $\cnt_\ell$, where $\cnt_{\ell}$ denotes the number of CCs of size $\ell$. 
		Initially, all the counters are set to  $0$ and all vertices are marked {\em unvisited}. We recursively choose an arbitrary unvisited vertex $v$, mark it  as {\em visited}  and start a BFS at $v$ which runs until (1) it has reached (e.g.~discovered an edge to) $1/\varepsilon +1$ unvisited vertices, (2) it reaches a visited vertex, or (3) the  BFS terminates because whole connected component (of size at most $1/\varepsilon$) containing $v$ has been explored. 
		Then we mark all the newly discovered vertices as {\em  visited} and update the counters accordingly. 
		More precisely, the static and dynamic algorithms are given in the following two tables.
		\begin{center}
			
			\begin{tabular}[h!]{|p{\textwidth}|}
				\hline
				\textsc{StaticNCC}($G_0,\varepsilon$) $\triangleright$	\textbf{A static algorithm for the number of CCs of size at most $1/\varepsilon$}%
				\begin{enumerate}
					\item Initialize $\cnt_\ell=0$, for each $1\leq \ell \leq 1/\varepsilon$. Mark all vertices as unvisited.
					\item While there exists some unvisited vertex $v$: 
					\begin{enumerate}
						\item Do BFS from $v$ until (i) $1/\varepsilon+1$ unvisited vertices have been reached, or (ii) any visited vertex has been reached, or (iii) no more new vertices can be reached. Mark all the newly discovered vertices in the search as visited.
						\item If (iii) occurs, and $\ell$ vertices have been reached for some $\ell \leq {1}/{\varepsilon}$, then increment $\cnt_\ell$ by $1$.
					\end{enumerate}
					\item Define the estimator  $\overline{c}:=\sum_{\ell=1}^{1/\varepsilon} \cnt_\ell$.
				\end{enumerate}
				
				\\ \hline
			\end{tabular}
		\end{center}
		
		Note that %
		by definition of the algorithm and the estimator $\overline{c}$, it holds that $\overline{c}=\nscc(G_0)$. 
		\begin{lemma}
			The algorithm \textsc{StaticNCC}($G,\varepsilon$) can be implemented in $O(n\cdot \frac{1}{\varepsilon})$ time for any $n$-vertex graph $G$. 
		\end{lemma}
		\begin{proof}
			Note that it suffices to bound the time of {\em exploring} each CC $C$, i.e., until all the vertices inside $C$ have been marked as visited. Note that $\cnt_\ell$ is exactly the number of CCs of size $\ell$, for $\ell\leq 1/\varepsilon$ and consider two cases, which together show the $O(n/\varepsilon)$ bound.
			\begin{enumerate}
				\item If $|C|=\ell \leq \frac{1}{\varepsilon}$, then the total time for exploring $C$ is $O(\ell^2)$. In this case, we note that the total time for exploring CCs of size at most $1/\varepsilon$ is $\sum_{\ell=1}^{1/\varepsilon} \cnt_\ell\cdot O(\ell^2)\leq \sum_{\ell=1}^{1/\varepsilon} \cnt_\ell\cdot \ell \cdot O(1/\varepsilon)= O(n/\varepsilon)$, where the last equation follows from the fact that $\sum_{\ell=1}^{1/\varepsilon}\cnt_\ell\cdot \ell\leq n$. 
				
				\item If $|C|>1/\varepsilon$, let $S=\{v_1,v_2,\cdots, v_b\}$ denote the set of vertices from which we start a BFS in $C$. Note that for each $i\leq b$, the number of newly discovered vertices from vertex $v_i$ is at most $1/\varepsilon+1$ by the description of our algorithm. 
				Let $t_j$ denote the number of vertices in $S$ from which the BFS discovers exactly $j$ new vertices, for each $j\leq 1/\varepsilon+1$. Then $|C|=\sum_{j=1}^{1/\varepsilon+1} t_j\cdot j$. 
				Furthermore, we note that for each $j\geq 1$, it takes time $O(j\cdot \frac{1}{\varepsilon})$ for the BFS to discover exactly $j$ new vertices, as we will only scan at most $\frac{1}{\varepsilon}+1$ neighbors for each of these new vertices. 
				Thus, the total time of exploring $C$ is  $\sum_{j=1}^{1/\varepsilon+1} t_j\cdot O(j\cdot \frac{1}{\varepsilon})\leq O(1/\varepsilon) \cdot \sum_{j=1}^{1/\varepsilon+1} t_j\cdot j =O(|C|/\varepsilon)$. 
				Thus, the total time of exploring CCs of size at least $1/\varepsilon+1$ is $\sum_{C: |C|\geq 1/\varepsilon+1} O(|C|/\varepsilon)=O(n/\varepsilon)$, where the last equation follows from the fact that $\sum_{C:|C|\geq 1/\varepsilon + 1} |C| \leq n$.  	
			\end{enumerate}
		\end{proof}

		\paragraph{Handling edge updates} Now we give the details of our dynamic algorithm \textsc{DeterDynamicNCC}($G,\varepsilon$) %
		for updating the counter $\overline{c}$ by running a limited BFS from the two endpoints of the updated edge in the graph before and after the update. %
		\begin{figure}[h!]
			\begin{center}
				\begin{tabular}{|p{\textwidth}|}
					\hline
					\textsc{DeterDynamicNCC}($G,\varepsilon$) $\triangleright$			\textbf{Maintaining an estimator for $\ncc(G)$ of a dynamic graph $G$}
					\begin{enumerate}	
						\item Preprocessing: run the algorithm 		\textsc{StaticNCC}($G_0,\varepsilon$) to find the $\overline{c}$, the number of CCs of $G_0$ of size at most $1/\varepsilon$.

						\item Handling an edge insertion $(u,v)$: perform three BFS calls: two from $u$ and $v$, respectively, in the graph before the insertion of $(u,v)$, and one from $u$ in the graph after 
						the insertion. 
						Stop the BFS once $1/\varepsilon+1$ vertices have been reached or no more new vertices can be reached. 
						Let $s_u^{(0)},s_v^{(0)},s_u^{(1)}$ denote the sizes of the corresponding explored subgraphs.		 
						\begin{enumerate}
							\item\label{alg:one_large} If exactly one of $s_u^{(0)}$ and $s_v^{(0)}$, say $s_u^{(0)}$, is no larger than $1/\varepsilon$, then decrement $\overline{c}$ by $1$.
							
							\item\label{alg:two_small} If both of $s_u^{(0)},s_v^{(0)}$ are smaller than $1/\varepsilon$:
							\begin{enumerate}
								\item\label{alg:two_small_1}  if $s_u^{(1)}$ is larger than $1/\varepsilon$, then decrement $\overline{c}$ by $2$; 
								\item\label{alg:two_small_2} if $s_u^{(1)}$ is no larger than $1/\varepsilon$ and $s_u^{(1)}=s_u^{(0)}$, then keep $\overline{c}$ unchanged;
								
								\item\label{alg:two_small_3} if $s_u^{(1)}$ is no larger than $1/\varepsilon$ and $s_u^{(1)}\neq s_u^{(0)}$, then decrement $\overline{c}$ by $1$.
							\end{enumerate}
						\end{enumerate} %
						\item Handling an edge deletion $(u,v)$: perform three BFS calls: one from $u$ in the graph before the deletion of $(u,v)$, and two from $u$ and $v$, respectively, in the graph after 
						the deletion. 
						Stop the BFS once $1/\varepsilon+1$ vertices have been reached or no more new vertices can be reached. 
						Let $s_u^{(0)},s_u^{(1)},s_v^{(1)}$ denote the sizes of the corresponding explored subgraphs.	
						\begin{enumerate}
							\item If exactly one of $s_u^{(1)}$ and $s_v^{(1)}$, say $s_u^{(1)}$, is no larger than $1/\varepsilon$, then increment $\overline{c}$ by $1$.
							
							\item If both of $s_u^{(1)},s_v^{(1)}$ are smaller than $1/\varepsilon$:
							\begin{enumerate}
								\item  if $s_u^{(0)}$ is larger than $1/\varepsilon$, then increment $\overline{c}$ by $2$; 
								\item if $s_u^{(0)}$ is no larger than $1/\varepsilon$ and $s_u^{(0)}=s_u^{(0)}$, then keep $\overline{c}$ unchanged;
								
								\item if $s_u^{(0)}$ is no larger than $1/\varepsilon$ and $s_u^{(0)}\neq s_u^{(1)}$, then increment $\overline{c}$ by $1$.
							\end{enumerate}
						\end{enumerate}
					\end{enumerate}		
					\\ \hline
				\end{tabular}
			\end{center}

		\end{figure}
		{\bf Correctness.} As aforementioned, it suffices to show that throughout all the updates, $\bar{c}=\nscc(G)$, where $G$ is the current graph. We give details as follows.
		
		For any edge insertion $(u,v)$, we know that the number $\nscc$ (of CCs of size at most $1/\varepsilon$) can change by at most $2$. More precisely, it changes if and only if at least one of $s_u^{(0)},s_v^{(0)}$ is at most $1/\varepsilon$ and $u,v$ do not belong to the same CC before the edge insertion, where $s_u^{(0)}$ and $s_v^{(0)}$ are the sizes of the explored subgraphs (before the edge insertion) starting from $u$ and $v$, respectively, that we compute in the algorithm. 
		Furthermore, if Step~\ref{alg:one_large} happens, i.e., exactly one of $s_u^{(0)}$ and $s_v^{(0)}$, say $s_u^{(0)}$, is no larger than $1/\varepsilon$, then a small CC merges into a large CC, and thus $\nscc$ decreases by $1$. If Step~\ref{alg:two_small} happens (i.e., $s_u^{(0)},s_v^{(0)}$ are smaller than $1/\varepsilon$): 
		if Step~\ref{alg:two_small_1} happens, i.e., $s_u^{(1)}$ is larger than $1/\varepsilon$, then two small CCs merge into a CC of size larger than $1/\varepsilon$ and thus $\nscc$ decreases by $2$; if Step~\ref{alg:two_small_2} happens, i.e., $s_u^{(1)}$ is no larger than $1/\varepsilon$ and $s_u^{(1)}=s_u^{(0)}$, then $u,v$ belong to the same CC before $(u,v)$ was inserted and thus $\nscc$ remains unchanged; if Step~\ref{alg:two_small_3} happens, i.e., $s_u^{(1)}$ is no larger than $1/\varepsilon$ and $s_u^{(1)}\neq s_u^{(0)}$, then two small CCs merge into a CC of size no larger than $1/\varepsilon$ and thus $\nscc$ decreases by $1$.
		By the description of our algorithm, after the insertion $(u,v)$, the maintained $\overline{c}$ still satisfies that $\overline{c}=\nscc(G')$, where $G'$ is the updated graph. 
		
		The case for edge deletions can be analyzed similarly.

		{\bf Running time.} Now we analyze the running time of our dynamic algorithm. %
		We note that for each update (either insertion or deletion), we only need to execute $O(1)$ BFS calls, each of which will explore at most $O(1/\varepsilon)$ vertices (and thus $O(1/\varepsilon^2)$ edges). Thus, the worst-case update time is $O(1/\varepsilon^2)$.
		
		Theorem \ref{thm:deterministic_NCC} now follows from the above analysis of the static algorithm \textsc{StaticNCC} and the dynamic algorithm \textsc{DeterDynamicNCC}. 
	\end{proof}

	\section{A Randomized Dynamic Algorithm}\label{subsec:randomized}

	Now we present our randomized dynamic algorithm $\mathcal{A}$ for estimating the weight of an MSF of $G$ and prove Theorem \ref{thm:random_WMST}. To do so, as  for the previous deterministic algorithm, we give dynamic algorithms $\mathcal{B}_j$ that approximate the number of CCs of subgraphs $G^{(j)}$ (as defined before) of $G$ within an additive error $O(\varepsilon \nis(G)/W)$ with small worst-case update time. 

	We first consider a dynamic algorithm whose input is a dynamic graph $H$ and an additional (integer) parameter $T$ that is also changing along with each edge update. That is, starting  with an initial graph and parameter, at each timestamp $i\geq 1$, the graph $H$ is updated by one edge insertion or deletion, and the parameter $T$ is updated to a new value. We give such a dynamic algorithm for estimating the number of CCs of $H$ with an additive error parametrized by $T$.} We will show the following theorem.

\begin{restatable}{theorem}{randomNCC}\label{thm:random_NCC}
	Let $\varepsilon', p\in (0,1)$. Consider a dynamic algorithm that takes as input a dynamic graph $H$ and an additional (integer) parameter $T$, called $T$-parameter, that is also changing along with each edge update. Suppose that (a) $\Thr\geq \nis(H)$ throughout all the updates and (b) the $T$-parameters at timestamps $i$ and $i-1$ differ by at most $2$, for $i\geq 1$. 
	Then there exists a fully dynamic algorithm {\textsc{RandDynamicNCC}($H$, $\Thr$, $\varepsilon'$, $p$)} that with probability at least $1-p$, maintains an estimator $\overline{cc}$ for the number $\ncc$ of CCs of a dynamic graph $H$ such that  $|\overline{cc}-\ncc(H)|\leq \varepsilon'\cdot \Thr$. The worst-case time per update operation is $O(\max\{1,{\log(1/\varepsilon')\log(1/p)}/(\varepsilon'^3 \Thr^{*})\})$, where $\Thr^{*}$ is the minimum of the $T$-parameters over all updates. Our algorithm works against an adaptive adversary.
\end{restatable}

We defer the proof of Theorem \ref{thm:random_NCC} to the end of this section. To approximate the weight of an MSF of a dynamic graph $G$, our algorithm $\mathcal{B}_j$ for approximating the number of CCs of $G^{(j)}$ is obtained by invoking the algorithm \textsc{RandDynamicNCC} from Theorem \ref{thm:random_NCC} on $G^{(j)}$ with the $T$-parameter being specified to be the number $\nis(G)$ of  non-isolated vertices in $G$. 
Note that by our choice of $T$-parameter and choosing $p=1/\poly(n)$, the algorithm $\mathcal{B}_j$ approximates the number of CCs of each subgraph $G^{(j)}$ within an additive $\varepsilon' \nis(G)$, with worst-case $O(\max\{1,\frac{\log(1/\varepsilon')\log(n)}{\varepsilon'^3 \nis^{*}}\})$ time per update operation, with high probability, where $\nis^*$ is the minimum number of non-isolated vertices in $G$ throughout all the updates. The running time can be further guaranteed to be $O(\max\{1,{\log(1/\varepsilon')\log(n)}/(\varepsilon'^3 \sqrt{m^*})\})$ as it holds that $\nis^*>\sqrt{2m^*}$ (a fact that is proven in the proof of Theorem \ref{thm:random_WMST}).

By combining the dynamic algorithm for the number of CCs given in Theorem~\ref{thm:random_NCC} and the relation of the weight of MSF of a graph $G$ and the number of CCs in its subgraphs as given in Lemma~\ref{lem:mst_ncc}, we are ready to present our randomized algorithm \textsc{RandEstWMSF} for the weight of an MST of $G$. \pnew{For each $1\leq j\leq r=\lceil\log_{1+{\varepsilon/2}} W\rceil$, it first initializes a data structure  $\mathcal{B}_j$ for maintaining an estimator $\overline{c}_j$ of the number of CCs of each $G^{(j)}$ using the initial graph $G_0$, $\ncc(G_0^{(j)})$,   $T=\nis(G_0)$, and corresponding parameters $\varepsilon'=\varepsilon/(24W),p=p'/r$. Then for each update in $G$, the algorithm first updates each $G^{(j)}$ accordingly and then invoke $\mathcal{B}_j$ with $T$-parameter $\mathrm{nis}(G)$ to update $\overline{c}_j$. To guarantee that the $T$-parameter for $G^{(j)}$ fulfills the precondition of Theorem \ref{thm:random_NCC}, we introduce the use of self-loops to  ``synchronize'' the updates in $G$ and $G^{(j)}$, for each $j\leq r$. That is, for each update on the graph $G$, it is guaranteed that there will be one or two corresponding updates on each $G^{(j)}$, for each $j\leq r$ (see Footnote \ref{footnote:selfloop} for more explanations). The algorithm  \textsc{RandEstWMSF} is formally described in the following table. }

\begin{table}[h]

\begin{center}
	
	\begin{tabular}{|p{\textwidth}|}
		\hline
		\textsc{RandEstWMSF}($G,\varepsilon, p',W$) $\triangleright$ \textbf{Dynamically Maintaining an estimator $\overline{M}$ of the weight of an MSF}%
		\begin{enumerate}
			\item Let $r=\lceil\log_{1+{\varepsilon/2}} W\rceil$, $\ell_j=(1+{\varepsilon/2})^j$, and $\lambda_j=(1+{\varepsilon/2})^{j+1}-(1+{\varepsilon/2})^j$.
			\item Let $G^{(j)}$ denote the subgraph of the current graph $G$ spanned by all edges with weights at most $\ell_j$. Let $c^{(j)}$ be the number of CCs of $G^{(j)}$.
			\item \pnew{For each $1\leq j\leq r$, initialize an estimator $\overline{c}_j$ for $c^{(j)}$ to be $\overline{c}_j=\ncc(G_0^{(j)})$, where $G_0$ is the initial graph of $G$; initialize a data structure $\mathcal{B}_j$ for  \textsc{RandDynamicNCC}($H$, $T$, $\varepsilon'$, $p$) from Theorem~\ref{thm:random_NCC} with $p={p' }/{r}$, $\varepsilon' = {\varepsilon}/{(24W)}$, $H=G_{0}^{(j)}$ and $T=\nis(G_0)$. }
			\item \pnew{For each  $i\geq 1$ and the $i$-th update in $G$: }
			
			\pnew{for each $j$ with $1\leq j\leq r$:}
			\begin{enumerate}
				\item\label{alg:good} if the $i$-th update in $G$ leads to an update in $G^{(j)}$, perform the same update to $H$ in $\mathcal{B}_j$, and set the corresponding $T$-parameter to be $\nis(G)$, where $G$ is the graph after the $i$-th update; 
				\item if the $i$-th update with edge endpoints $u,v$ in $G$ does not lead to an update in $G^{(j)}$, perform two consecutive updates to $H$ in $\mathcal{B}_j$: %
				\begin{enumerate}
					\item\label{alg:insert} insert self-loop\footnote{\label{footnote:selfloop}
						\pnew{Note that if we do not introduce the self-loops, then it is not the case that each edge update in $G$ leads to an update in every graph $G^{(j)}$. In particular, if the updated edge of $G$ has weight larger than $\ell_j$, then the graphs $G^{(i)}$ with $\ell_i<\ell_j$ stay the same, i.e., they are \emph{not} updated. This further leads to the issue that when one edge update occurs on $G^{(i)}$, the $T$-parameter for $G^{(i)}$ may have changed by more than $2$, i.e., the precondition of Theorem \ref{thm:random_NCC} is no longer fulfilled. This could happen if before this update in $G^{(i)}$, $G$ was updated many times (where each updated edge had weight more than $\ell_i$) and thus the $T$ parameter (i.e., $\nis(G)$) of $G^{(i)}$  has changed by more than $2$.}} $(u,u)$ to $H$, and set the corresponding $T$-parameter to be $\nis(G')$, where $G'$ is the union of the graph $G$ after the $i$-th update together with the self-loop $(u,u)$;
					\item\label{alg:delete} delete self-loop $(u,u)$ from $H$, and set the corresponding $T$-parameter to be $\nis(G)$, where $G$ is the graph after the $i$-th update. 
					
				\end{enumerate} 
				\item \pnew{update  $\mathcal{B}_j$ and $\overline{c}_j$ according to the above edge update(s) in $H$.}
			\end{enumerate}
			\item Define $\overline{M}:=n-\overline{c}_r\cdot (1+{\varepsilon}/{2})^r+\sum_{j=0}^{r-1}\lambda_j\cdot \overline{c}_j$.
		\end{enumerate}
		
		\\ \hline
	\end{tabular}
\end{center}
\end{table}
Now we restate Theorem \ref{thm:random_WMST} and give its proof.

\randomWMSF*
\begin{proof}

	Let $p'$ be a parameter that will be specified later. Let $r,\ell_j, \lambda_j, G^{(j)}, c^{(j)}, \overline{M}$ be defined as in the algorithm \textsc{RandEstWMSF}($G$, $\varepsilon$, $p'$,$W$). Note that $\nis(G)\ge \nis(G^{(j)})$, since $G^{(j)}$ is a subgraph of $G$ for any $1 \le j \le r$.
	Since $G$ is simple, we know that $M\geq {\nis(G)}/{2}$, as each non-isolated vertex is incident
	to at least one edge (of weight at least $1$) of any MSF.

	Note that $\overline{c}_j$ is the estimator of $c^{(j)}$ obtained by invoking the dynamic algorithm in Theorem~\ref{thm:random_NCC} with $p={p' }/{r}$, $\varepsilon' = {\varepsilon}/{(24W)}$ and the above specified dynamic graph $H$ and $T$-parameter. 
	As each update in $G$ changes $\nis(G)$ by at most $2$, $T$-parameter changes by at most $2$ for each update in $G^{(j)}$ (and thus fulfills requirement (b) of Theorem~\ref{thm:random_NCC}) {\em if} we guarantee that each update corresponding to edge $(u,v)$ in $G$ leads to at least one update in
	$G^{(j)}$. %
	In the case that the weight of $(u,v)$ is at most $\ell_j$, then 
	$(u,v)$ belongs to $G^{(j)}$, and, thus, the requirement is naturally fulfilled.

	If the weight of $(u,v)$ is larger than $\ell_j$, we execute two
	operations that cancel each other. Specifically we insert a self-loop $(u,u)$ at Step \ref{alg:insert} and then immediately delete it at Step \ref{alg:delete}, so that the graph $G^{(j)}$
	remains correct, but the $T$-parameters in these two steps also fulfill requirements (a) and (b). 
	More precisely, after adding the self-loop $(u,u)$, it holds that the $T$-parameter is at least $\nis(H)$ as $H$ is the number of non-isolated vertices of a subgraph of $G$ with the self-loop $(u,u)$; furthermore, the change of the $T$-parameter is at most $1$ as the self-loop can change the number of non-isolated vertices by at most $1$, so 
	the requirement (b) holds for this operation. After deleting the self-loop $(u,u)$, the requirement (a) is still guaranteed, as $H$ will be exactly $G^{(j)}$, which implies that the corresponding $T$-parameter is at least $\nis(G)\geq \nis(H)$; the requirement (b) still holds as deleting one self-loop can change the $T$-parameter by at most $1$. 

	For simplicity, we slightly abuse the notation by letting $T_i$ denote the specified $T$-parameter in the algorithm corresponding to the $i$-th update of $G$ (i.e., it corresponds to the $T$-parameter in either Step \ref{alg:good}, \ref{alg:insert} or \ref{alg:delete}).  By~Theorem~\ref{thm:random_NCC} and the fact that $\ncc(H)=\ncc(G^{(j)})$ (as $H$ is either $G^{(j)}$ itself or 
	$G^{(j)}$ with an additional self-loop $(u,u)$, for some vertex $u$), the algorithm computes an estimator $\overline{c}_j$ for $c^{(j)}$ such that with probability $1-{p' }/{r}$, it holds that 
	\[\abs{\overline{c}_j-c^{(j)}}\leq \varepsilon'\cdot \Thr_i \leq {\varepsilon}/{(24W)} \cdot (\nis(G)+1)\leq {\varepsilon}/{(12W)} \cdot \nis(G),\] %
	where the second inequality holds as $T_i\leq \nis(G)+1$ throughout the update, and the last inequality follows as $\nis(G)\geq \nis^*\geq 1$.

	Thus with probability at least $1-r\cdot \frac{p'}{r}=1-p'$, the above inequality holds for all $j$ such that $1\leq j\leq r$. Assume in the following that this event holds. Let $X=n-c^{(r)}\cdot (1+{\varepsilon}/{2})^r+\sum_{j=0}^{r-1}\lambda_i\cdot c^{(j)}$. By the same calculation as in the proof of Theorem \ref{thm:deterministic_WMST}, we have that 
	\[ |\overline{M}-X| \leq  {\varepsilon \nis(G)}/{4}\leq {\varepsilon M}/{2} \]
	
	Now by Lemma \ref{lem:mst_ncc} with approximation parameter ${\varepsilon}/{2}$, it holds that $M\leq X\leq (1+\varepsilon/2)M$. Thus $\overline{M}$ is a $(1+\varepsilon)$-approximation of $M$.

	Now note that for each $j$, it always holds that for the $i$-th update in $G$ satisfies that $T_i\geq \nis(G)\geq \nis^*$, and thus the worst-case time spent per update for computing $\overline{c}_j$ is
	\[O(\max\{1,{\log(1/\varepsilon')\log(1/p)}/(\varepsilon'^3 \nis^{*})\})=O(1+ {W^3\log(W/\varepsilon)\log(r/p')}/(\varepsilon^3 \nis^*)).
	\] %
	In total, the worst-case time per update operation of \textsc{RandEstWMSF}($G$, $\varepsilon$, $p'$,$W$) is 
	\begin{align*}
		&O(\sum_{j=1}^{r}(1+ {W^3\log(W/\varepsilon )\log(r/p')}/({\varepsilon^3 \nis^*}) )) \\
		=&\pnew{O((\log W)/\varepsilon+{W^3\log W\log({W}/{\varepsilon} )\log({\log W}/(\varepsilon p'))}/(\varepsilon^4 \nis^*))}.
	\end{align*}

	Now we note that $\nis^*> \sqrt{2 m^*}$, where $m^*$ is the minimum number of edges of the graph throughout all the updates. This is true as for the graph $G^*$ with minimum non-isolated vertices, i.e., $\nis(G^*)=\nis^*$, the total number of edges in $G^*$ is both at most $(\nis^*)(\nis^*-1)/2<(\nis^*)^2/2$ and also at least $m^*$.  
	
	Therefore, by setting $p'=1/n^c$ for any constant $c>1$, we can guarantee that the worst-case update time of the algorithm \textsc{RandEstWMSF}($G$, $\varepsilon$, $p'$,$W$) is 
	\[\pnew{O((\log W)/\varepsilon+{W^3(\log W)\log({W}/{\varepsilon} )(\log({\log W}/\varepsilon) +\log n)}/({\varepsilon^4 \sqrt{m^*}})),}\]
	and the algorithm succeeds with probability at least $1-p'=1-1/n^c$.

	The algorithm works against an adaptive adversary as each of the algorithms from Theorem~\ref{thm:random_NCC} works against an adaptive adversary, and the MSF algorithm
	returns only a weighted sum of the values returned by each of these algorithms.
\end{proof}

\subsection{Proof of Theorem \ref{thm:random_NCC}}\label{sec:dynamic_thr}
In the following, we present the proof of Theorem \ref{thm:random_NCC}. We use $H$ to denote the considered graph. %
\randomNCC*
\subsubsection{A data structure for dynamically sampling non-zero entries}\label{sec:nonzero}
We first give a data structure to maintain some data structures so that the algorithm can perform \emph{non-isolated vertex-sample queries}, i.e., to sample a vertex from the set of all non-isolated vertices uniformly at random \emph{at any time}. Our dynamic algorithm for estimating the weight of an MSF will make use of this data structure.

We first present a more general data structure to sample non-zero entries from a dynamic array and then show how to use it to support non-isolated vertex-sample queries. Given a set $V$ of  $n$ elements (here vertices), numbered from $0$ to $n-1$, each element $u$ with an associated number $d_u$ (here degree), we show how to support the following operations in constant time with preprocessing time $O(n)$:
\begin{enumerate}
	\item[--]{\bf Update($u,\delta$):} add $\delta$ to $d_u$, where $\delta$ can be positive or negative.
	\item[--]{\bf Non-zero sample():} return an element that is chosen uniformly at random from all elements $u$ with $d_u \not = 0$.
\end{enumerate}
\begin{lemma}\label{lemma:samplenonzero}
	There exists a data structure that supports {\bf Non-zero sample()} queries to a fully dynamic array with constant update time.
\end{lemma}
\begin{proof}
	Let us call an element $u$ of $V$ with $d_u \not= 0$ a {\em non-zero} element.
	We implement the data structure by using two arrays and a counter:
	\begin{enumerate}
		\item Keep the number $\nis$ of non-zero elements of $V$.
		\item Keep an array $\mathcal{A}$ of size $n$, where only the first $\nis$  entries are used, such that (i) each entry in $\mathcal{A}$ stores a non-zero element $u$  together with $d_u$ and (ii) each non-zero element of $V$ is stored in $\mathcal{A}$ within the first $\nis$ entries.
		\item Keep an array $\mathcal{P}$ of size $n$, which has an entry for each element of $V$, such that if an element $u$ is stored in $\mathcal{A}[i]$ (i.e. $u$ is non-zero), then $\mathcal{P}[u] = i$; and if an element $u$ is not stored in $\mathcal{A}$ (i.e. $d_u = 0$), then $\mathcal{P}[u] = -1$. Thus $\mathcal{P}$ contains indices corresponding to the positions of elements in $\mathcal{A}$ or $-1$.
	\end{enumerate}
	During preprocessing we initialize both arrays, set all entries of $\mathcal{P}$ to $-1$, and set $\nis=0$. Then insert each $u$ whose initial value $d_u \not= 0$ by calling \textbf{Update}($u, d_u$).

	\subparagraph{Handling an \textbf{Update}($u, \delta$) operation.} Whenever an {\bf Update}($u, \delta$) operation is executed, we check if $\mathcal{P}[u] > -1$. 
	\begin{enumerate}
		\item
		If $\mathcal{P}[u] > -1$, then $u$ is stored in $\mathcal{A}$ and $\mathcal{P}[u]$ contains the index of $u$ in $\mathcal{A}$.
		Thus we add $\delta$ to $d_u$, which is retrieved and then stored in the entry $\mathcal{A}[\mathcal{P}[u]]$. If the resulting value $d_u \not= 0$, this completes the update operation.
		If, however, the resulting value $d_u = 0$, let $v$ be the element stored in $\mathcal{A}[\nis]$.
		We copy into
		$\mathcal{A}[\mathcal{P}[u]]$ all information of element $v$, which we retrieve from $\mathcal{A}[\nis]$. Then we set
		$\mathcal{P}[v] = \mathcal{P}[u]$, set $\mathcal{P}[u] = -1$, and decrement $\nis$.
		\item If $\mathcal{P}[u] =-1$, then we increment $\nis$ by 1, set $d_u = \delta$, store $u$ and $d_u$ in $\mathcal{A}[\nis]$, and set $\mathcal{P}[u] = \nis$.
	\end{enumerate}

	\subparagraph{Handling \textbf{Non-zero sample} operation.} To implement this operation, we pick a random integer $j$ between $0$ and $\nis-1$ and return the element from $\mathcal{A}[j]$.
\end{proof}

\textbf{Supporting non-isolated vertex-sample queries in dynamic graphs.} Next we show how to use the above data structure to support non-isolated vertex-sample query throughout all the updates.
Whenever an edge $(u,v)$ is inserted, for each $x \in \{u,v\}$, we call \textbf{Update}($x, 1$).
Whenever an edge $(u,v)$ is deleted, for each $x \in \{u,v\}$, we call \textbf{Update}($x, -1$).
To sample a non-isolated vertex, we call \textbf{Non-zero sample}().

\subsubsection{The dynamic algorithm}

\paragraph{A static algorithm for estimating the number of CCs}%
We will make use of a static algorithm, which we briefly describe here. Recall that $\ncc(H)$ denotes the number of CCs of a graph $H$. The following lemma was proven by Berenbrink et al.~\cite{BKM14:numcc} (which improves upon the result in~\cite{CRT05:MST}) and it gives a constant-time algorithm for estimating $\ncc(H)$. %
It is assumed that the algorithm can perform a \textit{vertex-sample query}, which allows it to sample (in constant time) a vertex uniformly at random from $V$, and can make queries to the adjacency list of the graph. Note that these two queries for accessing a static graph can be supported by maintaining an array of vertices and the adjacency list of the graph, respectively.%
\begin{lemma}[\cite{BKM14:numcc}]\label{lemma:bkm_ncc}
	Let $\varepsilon >0$ and $0<p<1$. Suppose the algorithm has access to the adjacency list of a graph $H$ and can perform vertex-sample queries. Then there exists an algorithm that with probability at least $1-p$, returns an estimate that approximates $\ncc(H)$ with an additive error $\varepsilon n$. The running time of the algorithm is $O(1/\varepsilon^2\log(1/\varepsilon)\log(1/p))$. 
\end{lemma}

We remark that the algorithm in~\cite{BKM14:numcc} simply samples (uniformly at random) $O(1/\varepsilon^2)$ vertices,
performs a BFS starting from each sampled vertex (for a number of steps) and then makes decisions based on the explored subgraphs. %
{Let $H_{\nis}$ denote the subgraph induced by all non-isolated vertices in a graph $H$. 
	Note that if the algorithm is able to perform a {non-isolated vertex-sample} query, i.e., the algorithm can sample a vertex uniformly at random from the vertex set of $H_{\nis}$, then one can approximate the number $\ncc(H_{\nis})$ of CCs in the subgraph $H_{\nis}$ with an additive error $\varepsilon |V(H_{\nis})|=\varepsilon \nis(H)$, where $\nis(H)$ is the number of non-isolated vertices in $H$. This is true as we can simply treat $H[N]$ as the input graph in the algorithm from Lemma~\ref{lemma:bkm_ncc}. This also proves the following lemma, Lemma \ref{lemma:bkm_ncc_nis}, which will be invoked by our dynamic algorithm. %
	\begin{lemma}\label{lemma:bkm_ncc_nis}
		Let $\varepsilon >0$ and $0<p<1$. Suppose the algorithm has access to the adjacency list of a \emph{static} graph $H$ and can sample a vertex uniformly at random from the set of all non-isolated vertices in $H$. Then there exists an algorithm that with probability at least $1-p$, returns an estimate $\overline{b}$ that approximates $\ncc(H_{\nis})$ with an additive error $\varepsilon\cdot \nis(H)$. The running time of the algorithm is $O(1/\varepsilon^2\log(1/\varepsilon)\log(1/p))$. 
	\end{lemma}

	\paragraph{From static to dynamic}  
	The idea of our dynamic algorithm is as follows. %
	By slightly abusing notation, let $H_i$ denote the
	graph $H$ after the $i$-th update (and $H_0$ denotes the initial graph). Our algorithm first maintains %
	the number $\Gamma$ of non-isolated vertices in the current graph in the straightforward way: we can maintain the degree of each vertex, and count the number of vertices with non-zero degrees, which can be done in constant update time. 

	During initialization we set $\Psi=\Thr_0$ and $\overline{c}=\ncc(H_0)$, which is calculated by running a static BFS traversal on $H_0$. Then it
	partitions the updates (and the corresponding graphs) into \emph{phases}. For notational convenience, we let phase $\PP_0$ consist only  of the graph $H_0$. We let $H'_j$ be the graph corresponding to the last update in $\PP_j$, for any $j\geq 0$. At the beginning of each phase $\PP_j$ (for $j\geq 1$), we are given two parameters, an estimate $\overline{c}$ and $\Psi$, which correspond to estimator of the number of CCs and the number of non-isolated vertices in the graph $H_{j-1}'$ corresponding to the last update in $\PP_{j-1}$.  The phase $\PP_j$ consists of all graphs corresponding to the next $\frac{\varepsilon' \Psi}{4}$ updates. The parameters $\overline{c}$ and $\Psi$ remain unchanged (i.e., they are as fixed at the beginning of the phase) until the last update in $\PP_{j}$. 
	More precisely, let $\Thr'_j$ be the $\Thr$-parameter corresponding to the last update in $\PP_j$. 
	At the end of $\PP_j$, the algorithm sets $\Psi$ to $\Thr'_j$, runs the static algorithm from~Lemma~\ref{lemma:bkm_ncc_nis} to obtain an estimate $\overline{b}$ for $\ncc([H'_j]_{\nis})$ and re-set the estimate $\overline{c}$ to $\overline{b}+n-\nis(H'_j)$.
	Then we start with a new phase $\PP_{j+1}$ consisting of the graphs corresponding to the next $\varepsilon' \Psi/4$ updates, where $\Psi$ is as fixed at the beginning of this phase and then we repeat the above. 
	
	Throughout, we maintain the adjacency list of the dynamic graph in a trivial way and maintain the data structure from Section \ref{sec:nonzero}
	for sampling a vertex uniformly at random from the set of all non-isolated vertices from the current graph. %
	When asked a query on the number of CCs of the current graph, the algorithm returns $\overline{c}$. 
	The description of our algorithm \textsc{RandDynamicNCC} is given as follows, omitting the details for maintaining adjacency list, array of degrees, and data structures for sampling non-isolated vertices.
	\begin{center}
		\begin{tabular}{|p{\textwidth}|}
			\hline
			\textsc{RandDynamicNCC}($H,T,\varepsilon',p$) $\triangleright$		\textbf{Maintaining an estimator $\overline{c}$ for $\ncc(H)$ of a dynamic graph $H$ (with parameter $\Thr \geq \nis(H)$) with additive error $\leq \varepsilon'\cdot \Thr$}
			\begin{enumerate}	
				\item Preprocessing: Use BFS in the initial graph $H_0$  to obtain $\nis(H_0)$ and $\ncc(H_0)$. Start of the first
				phase. Initialize $\Gamma=\nis(H_0)$ and $\overline{c}=\ncc(H_0)$, $\Psi=\Thr_0$. Let $i=1$.

				\item For the $i$-th update: %
				\begin{enumerate}
					\item update $\Gamma$ to be $\nis(H_i)$
					
					\item if $i \mod (\frac{\varepsilon'\cdot\Psi}{4}) = 0$, then \hspace{0.3cm} $\triangleright$ The end of a phase
					\begin{enumerate}
						\item compute an estimator $\overline{b}$ for $\ncc([H_{i}]_{\nis})$ by running the static algorithm in Lemma~\ref{lemma:bkm_ncc_nis} on $H_i$ with parameter $\varepsilon=\frac{\varepsilon'}{4}$ and $p$
						\item set $\overline{c}=\overline{b}+n-\Gamma$ \hspace{0.3cm} $\triangleright$ $n - \Gamma$ is the number of isolated nodes in $H_i$ %
						\item set $\Psi=\Thr_i$
					\end{enumerate}
					\item set $i=i+1$
				\end{enumerate} %
				
			\end{enumerate}		
			\\ \hline
		\end{tabular}
	\end{center}

	\subsubsection{Putting it together}
	Now we prove the correctness of the above dynamic algorithm, analyze its running time and then finish the proof of Theorem \ref{thm:random_NCC}.

	\subparagraph{Correctness.}
	Let $j\geq 0$ and let $\PP_j$ denote the set of all the graphs in the $j$-th phase in the algorithm. In particular, $\PP_0=\{H_0\}$. Let $H'_j$ and $\Thr'_j$ be the graph and the $\Thr$-parameter, respectively, corresponding to the last update in $\PP_j$. We first prove the following claims.
	
	\begin{claim}
		For any $j\geq 0$ and the graph $H=H'_j$, our estimate satisfies that $\abs{\overline{c}-\ncc(H'_j)}\leq {\varepsilon' \nis(H'_j)}/{4}\leq %
		{\varepsilon' \cdot \Thr'_j}/{4}$ with probability $1-p$. 
	\end{claim}
	
	\begin{proof}
		For $j=0$, it holds that $H'_0=H_0$ and our estimate $\overline{c}=\ncc(H_0)$ by definition. Let $j\geq 1$. Our algorithm calls the static algorithm from Lemma \ref{lemma:bkm_ncc_nis}, which returns with probability $1-p$
		an estimator $\overline{b}$ for $\ncc([H'_j]_{\nis})$
		such that $|\overline{b}-\ncc([H'_j]_{\nis})|\leq \frac{\varepsilon' \nis(H'_j)}{4}$. %
		This gives 
		$|\overline{c}-\ncc(H'_j)|=|\overline{b}+n-\nis(H'_j)-\ncc(H'_j)|=|\overline{b}-\ncc([H'_j]_{\nis})|\leq \frac{\varepsilon' \nis(H'_j)}{4}\leq \frac{\varepsilon' \Thr'_j}{4},$
		where the second equation follows from the fact that $\ncc(H'_j)$ is the sum of $\ncc([H'_j]_{\nis})$ and the number of isolated vertices, $n-\nis(H'_j)$.
	\end{proof}

	\begin{claim}
		For any $j\geq 0$, any $H\in \PP_{j+1}\setminus \{H'_{j+1}\}$ that is the graph right after the $i$-th update, it holds that $\abs{\overline{c}-\ncc(H)}\leq \varepsilon' \Thr_i$ with probability $1-p$.
	\end{claim}
	
	\begin{proof}
		Note that after update $f_j$, $\Psi=\Thr'_j\geq \nis(H'_j)$; both $\Psi$ and the estimate $\overline{c}$ remain unchanged before the last update in $\PP_{j+1}$. 
		As each update changes $\ncc(H)$ by at most 1,
		with at least probability $1-p$, it holds that $\abs{\overline{c}-\ncc(H)}\leq {\varepsilon' \nis(H'_j)}/{4}+{\varepsilon' \Psi}/{4} \leq {\varepsilon' \Psi}/{2}$ for all graph $H\in \PP_{j+1}\setminus\{H'_{j+1}\}$.
		Note that $|\Psi-\Thr_i|\leq \frac{\varepsilon'\Psi}{2}$, as each update (during the next $\frac{\varepsilon' \Psi}{4}-1$ updates) can change $\Thr_i$ by at most $2$. Thus $\Psi\leq \frac{1}{1-\varepsilon'/2}\Thr_i\leq (1+\varepsilon')\Thr_i\leq 2\Thr_i$. This implies that $\overline{c}$ approximates $\ncc(H)$ with an additive error $\varepsilon' \Thr_i$.
	\end{proof}

	Thus with probability $1-p$, for any $H$ (right after the $i$-th update), we have $\abs{\overline{c}-\ncc(H)}\leq \varepsilon' \Thr_i$. Finally, note that the algorithm uses ``fresh'' random bits at the beginning of each phase, only needs to access to the current graph, and does not reuse any information computed in prior phases. Within each phase we performed a worst-case analysis, i.e., we
	assumed that the adversary changes the graph in the worst possible way, i.e., changing $\ncc(H)$ by 1 in each update.
	Thus, our algorithm works against an adaptive adversary.
	
	\subparagraph{Running time.}
	For each phase that contains all graphs in $\PP_{j+1}\setminus \{H'_{j+1}\}$ with parameter $\Psi=\Thr'_j$, the amortized time (over the phase) per update operation is 
	\[
	O(\max\{1,{(1/\varepsilon'^2)\log(1/\varepsilon')\log(1/p)}/({\varepsilon' \Psi})\}) = O(\max\{1,{\log(1/\varepsilon')\log(1/p)}/({\varepsilon'^3 \Psi})\}).
	\]
	(Note that we always need to use $O(1)$ time to update the adjacency list and other data structures so as to provide query access to the graph). Since $\Thr^{*}$ is the minimum value $\Thr_i$ over graphs throughout all the updates, then in any phase, $\Psi\geq \Thr^{*}$ and the amortized update time is
	$O(\max\{1,\frac{\log(1/\varepsilon')\log(1/p)}{\varepsilon'^3 \Thr^{*}}\}).$ The worst-case update time guarantee follows from the standard global rebuilding technique. %

	\section{Lower Bounds: Proofs of Theorem \ref{thm:deterministiclowerbound} and \ref{thm:lower}}\label{sec:lower}
	
	\subsection{A Lower Bound for Deterministic Data Structure}
	Our first lower bound is built upon a construction by Henzinger and Fredman \cite{henzinger1998lower}.
	The \emph{parity prefix sum} problem is defined as follows: Given an array
	$A[1], \dots, A[n]$ with entries from $\{0,1\}$, initialized with 0,
	build a data structure that executes an arbitrary sequence of the following operations:
	(1) $\Add(i)$: Set $A[i] = (A[i] + 1) \mod 2$; (2) $\Sum(i)$: return $\sum_{j = 1}^i A[i] \mod 2$. Fredman and Saks~\cite{FS89:cell} showed the following theorem. 
	
	\begin{theorem}[\cite{FS89:cell}]
		There is an amortized lower bound of $\Omega(\log n/(\log \log n +\log b))$ update time per operation for the parity prefix sum problem in the cell probe model with word size $b$.
	\end{theorem}
	Now we show the following.
	\begin{theorem}\label{thm:HFlower_2}
		Let $G$ be a dynamic $n$-vertex graph with edge weights in $[1,W]$. Any data structure that dynamically maintains the weight $M$ of an MSF of a graph $G$  within an additive error less than $W/2$ for any $W\geq 1$, or a multiplicative factor of  $(1+\varepsilon)$,  for any $W> 2\varepsilon n$ {must perform $\Omega(\log n/(\log \log n +\log b))$ cell probes, where each cell has size $b$.}
	\end{theorem}
	\begin{proof}
		The proof is analogous to the lower bound for dynamic connectivity in~\cite{henzinger1998lower}. Give a parity prefix sum problem build a graph with $n$ vertices, labeled from $1, \dots, n$,
		one vertex called $even$, and one vertex called $odd$. At any point in time the graph maintains the following invariant: It consists of two paths, namely an \emph{even} path containing all vertices $i$ (called ``even'' vertices) such that $\Sum(i)$ returns $0$, in order of their indices, and such that the vertex in the chain is connected to vertex $even$; and an analogous \emph{odd} path containing all ``odd'' vertices, connected to vertex $odd$.
		It is shown that each $\Add(i)$ only leads to a constant number of edge updates in the graph.
		To determine which edges need to be deleted and inserted, a Van-Emde-Boas priority queue is maintained, which adds a cost of $O(\log \log n)$ to each $\Add(i)$ operation.
		To answer a parity query $\Sum(i)$ return $1$ iff vertex $odd$ and vertex $i$ are connected.
		
		Now we show how to use this construction to give a lower bound for the dynamic approximate MSF weight problem. We maintain the same graph as above, and give each edge weight 1.
		Thus the weight of an MSF before and after every operation is $n-2$.
		To answer a $\Sum(i)$ query we insert an edge of weight $W$ from vertex $odd$ to vertex $i$.
		If the answer is $1$, then $i$ and $odd$ were already connected before the edge insertion and the weight of an MSF is unchanged. If the answer is 0, then $i$ and $odd$ were not connected before the edge insertion and, thus, the weight of an MSF is now $n-2 + W$.
		Thus any dynamic MSF algorithm with multiplicative actor of  $(1+\varepsilon)$,  for any $W> 2\varepsilon n$ and any dynamic MSF algorithm with additive error less than $W/2$ for any $W\geq 1$, can distinguish the two cases. Thus, they must perform $\Omega(\log n/(\log \log n +\log b))$ cell probes, where each cell has size $b$.
	\end{proof}
	
	Now we are ready to prove Theorem \ref{thm:deterministiclowerbound}.
	\thmdeterlower*
	\begin{proof}
		We let $G'$ be a graph with $n'$ vertices and edge weights in $[1,W]$, for $n'=\lfloor\frac{W}{3\varepsilon}\rfloor$. %
		By Theorem \ref{thm:HFlower_2}, any data structure that dynamically maintains the weight $M$ of an MSF of a graph $G'$ with $W> 2\varepsilon n'$ within an additive error less than $W/2$ for any $W\geq 1$, or a multiplicative factor of  $(1+\varepsilon)$  must perform $\Omega(\log n'/(\log \log n' +\log b))$ cell probes, where each cell has size $b$.
		
		Note that  $n'=(\log n)^{\omega_n(1)}/\varepsilon$. Now assume that we have a data structure $\mathcal{A}$ that dynamically maintains the $(1+\varepsilon)$-approximation of the weight of an MSF of an $n$-vertex graph. (The case that $\mathcal{A}$ has an additive error less than $W/2$ can be analyzed similarly.)  We use $\mathcal{A}$ to the graph $\bar{G}$ which consists of $G'$ and $n-n'$ isolated vertices. Note that the weight of an MSF of $\bar{G}$ is the same as the weight of an MSF of $G'$. Thus $\mathcal{A}$ maintains a $(1+\varepsilon)$-approximation of the weight of an MSF of $G'$. %
		
		By setting $b=\Theta(\log n)$, we know that $\mathcal{A}$ must perform $\Omega(\log n'/(\log \log n' +\log b))=\omega_n(1)$ cell probes, where each cell has size $O(\log n)$. %
	\end{proof}
	
	\subsection{A Lower Bound for Randomized Data Structure}
	We let $m^*(G)$ be the minimum number of edges of a dynamic graph $G$ throughout all the updates. 
	We first prove the following. %
	\begin{theorem}\label{thm:firstlowerbound}
		Let $G$ be a dynamic graph with edge weights in $[1,W]$. Any data structure that dynamically maintains the weight $M$ of an MSF of a graph $G$ with $W=\Omega(\varepsilon\cdot m^*(G))$ and $m^*(G)=\Theta(n)$ within an additive error less than $W/2$ for any $W\geq 1$, or a multiplicative factor of  $(1+\varepsilon)$, %
		with probability at least $1-\frac{1}{n^c}$ for some constant $c>0$, {must perform $\Omega((\log n)^2/b)$ cell probes, where each cell has size $b$.}
	\end{theorem}
	\begin{proof}%

		In~\cite{PatrascuD06} Patrascu and Demaine construct an $n$-node graph and show that there exists a 
		sequence $\mathcal S$ of $T$ edge insertion, edge deletion, and query operations such that any data structure for dynamic connectivity {must perform  
			$\Omega((\log n)^2/b)$ cell probes to process the sequence, where each cell has size $b$. This 
			shows that the amortized number of cell probes per operation is $\Omega((\log n)^2/b)$.}

		The graph $G$ in the proof of \cite{PatrascuD06} consists of a $\sqrt n \times \sqrt n$ grid, where each node in column $1$ has exactly 1
		edge to a node of column 2 and no other edges, each node in column $i$, with $1 < i < \sqrt n$ has exactly 1 edge to a node of column $i-1$ and 1 edge to a node of column $i+1$ and no other edges, and
		each node in column $\sqrt n$ has exactly 1 edge to a node of column $\sqrt n -1$ and no other edges.
		Thus, the graph consists of $\sqrt n$ paths of length $\sqrt n -1$ and the edges between column $i$ and $i+1$ for any $1 \le i < \sqrt n$ represent a permutation of the $\sqrt n$ rows.
		The sequence $\mathcal S$ consists of ``batches'' of $O(\sqrt n)$ edge updates, replacing the permutation of some column $i$ by a new permutation for column $i$. Between the batches of updates are ``batches'' of connectivity queries, each consisting of $\sqrt n$ connectivity queries
		and a parameter $1 \le k \le \sqrt{n}$, where the $j$-th query 
		for $1 \le j \le \sqrt n$ of each batch tests whether the $j$-th vertex of column 1 is connected with a specific vertex of column $k$. 
		
		We now show how to use this construction to give a lower bound for the dynamic approximate MSF problem.
		We add a new vertex $s$ and an edge between $s$ and every vertex in column 1
		and give weight 1 to every edge.
		We now show how to modify each connectivity
		query $(u,v)$ such that it consists of a constant number of edge updates and one query for the value of an MSF. 
		Thus, in the resulting sequence $\mathcal S'$ the number of query operations equals the number of
		query operations in $\mathcal S$ and the number of update operations is linear in the number of update and query operations in $\mathcal S$. 
		Thus the total number of operations in $\mathcal S'$ is only a constant factor larger
		than the number of operations in $\mathcal S$, which, together with the result of~\cite{PatrascuD06}, implies that the amortized number 
		of cell probes per operation is $\Omega((\log n)^2/b)$.
		
		We now show how to simulate a connectivity query($u,v$), where $u$ is in column 1 and $v$ is in column $k$ for some $1 \le k \le \sqrt{n}$. 
		To simulate a connectivity query($u,v$) we (1) remove the edge from $u$ to
		$s$, (2) add the edge $(u,v)$ with weight $W$  and then (3) ask a query for the weight $M$ of an MSF. Afterwards we undo the changes to $G$.
		Note that if $u$ and $v$ are connected in $G$ then the edge ($u,v)$ does not belong to the MSF (and the graph is disconnected), otherwise it does. 
		Furthermore, the value $M$ will be $n-2$ if $u$ and $v$ are connected, and
		$n-2+W$ if $u$ and $v$ are not connected.
		This implies that if our data structure maintains $M$ within an additive error less than $W/2$, for any $W\geq 1$, or within a multiplicative factor of $1+\varepsilon$ for $W> 2\varepsilon n$, then we can also test connectivity. Then the statement of the theorem follows from the lower bound for testing connectivity.  %
		Finally, we note that $m^*(G)= n-O(\sqrt{n})=\Theta(n)$ by construction of the graph and thus $W=\Omega(\varepsilon m^*(G))>2\varepsilon n$. 
		This finishes the proof.
	\end{proof}

	{
		Now we are ready to give the proof of Theorem \ref{thm:lower}.
		\thmlower*
		\begin{proof}
			The proof is similar to the proof of Theorem \ref{thm:deterministiclowerbound}.

			We let $G'$ be a graph with $n'$ vertices and edge weights in $[1,W]$, for $n'=\lfloor\frac{W}{3\varepsilon}\rfloor$. %
			By Theorem \ref{thm:firstlowerbound}, any data structure that dynamically maintains the weight $M$ of an MSF of a graph $G'$ with $W=\Omega(\varepsilon \cdot m^*(G'))$ and $m^*(G')=\Theta(n')$ within an additive error less than $W/2$ for any $W\geq 1$, or a multiplicative factor of  $(1+\varepsilon)$, with probability at least $1-\frac{1}{(n')^c}$, must perform $\Omega((\log n')^2/b)$ cell probes, where each cell has size $b$.
			
			Note that  $n'=\Theta(n^\alpha/\varepsilon)$ where $\alpha>0$ is a constant. Now assume that we have a data structure $\mathcal{A}$ that dynamically maintains the $(1+\varepsilon)$-approximation of the weight of an MSF of an $n$-vertex graph with probability at least $1-\frac{1}{n^{c'}}=1-\frac{1}{(n')^c}$ for some constant $c'>0$. (The case that $\mathcal{A}$ has an additive error less than $W/2$ can be analyzed similarly.)  We use $\mathcal{A}$ to the graph $\bar{G}$ which is consisted of $G'$ and $n-n'$ isolated vertices. Note that the weight of an MSF of $\bar{G}$ is the same as the weight of an MSF of $G'$. Thus $\mathcal{A}$ maintains a $(1+\varepsilon)$-approximation of the weight of an MSF of $G'$. Note that $m^*(G)=m^*(G')$, and that $W=\Omega(\varepsilon \cdot m^*(G'))=\Omega(\varepsilon \cdot m^*(G))$ and $m^*(G)=m^*(G')=\Theta(n')$.
			
			By setting $b=\Theta(\log n)$, we know that $\mathcal{A}$ must perform $\Omega((\log n')^2/b)=\Omega(\log^2 (W/\varepsilon)/b)=\Omega(\log n)$ cell probes, where each cell has size $O(\log n)$. %
		\end{proof}
	}

	%
	%
	\bibliographystyle{elsarticle-num}
	\bibliography{dynamic_constant}

	\appendix
	\renewcommand{\thesection}{\Alph{section}}

	\section{A Note on Dynamically Estimating the Number of CCs}\label{app:note_ncc}
	\paragraph{Estimating $\ncc(G)$ with an additive error $\varepsilon n^{O(1)}$} We note that similar to our previous algorithm (in Section~\ref{subsec:randomized}) for estimating $\ncc(G)$ with an additive error $\varepsilon \Thr(G)$, if we simply invoke the static algorithm from Lemma~\ref{lemma:bkm_ncc_nis} on the current graph $H=G$ with parameters $\nis(H)=n$ and $\ncc(H_{\nis})=\ncc(G)$ to obtain an estimator $\overline{cc}$  and re-compute the estimator every $\Theta(\varepsilon n)$ updates, then the corresponding algorithm always maintain an estimator for $\ncc(G)$ with an additive error $\varepsilon n$. That is, we have the following theorem.
	\begin{theorem}\label{thm:random_ncc_espn}
		Let $1>\varepsilon >0$ and $0<p<1$. There exists a fully dynamic algorithm that with probability at least $1-p$, maintains an estimator $\overline{cc}$ for the number $\ncc$ of CCs of a graph $G$ s.t.,  $|\overline{cc}-\ncc(G)|\leq \varepsilon\cdot n$. The worst-case time per update operation is $O(\max\{1,\frac{\log(1/\varepsilon)\log(1/p)}{\varepsilon^3 n}\})$.
	\end{theorem}
	
	The following is a direct corollary of the above theorem.
	\begin{corollary}\label{cor:random_NCC}
		Let $\varepsilon >0$ and let $c$ be any constant such that $c\geq 1$. There exists a fully dynamic algorithm that with probability at least $1-\frac{1}{n^c}$, maintains an estimator $\overline{cc}$ for the number $\ncc$ of CCs of a graph $G$ s.t.,  $|\overline{cc}-\ncc(G)|\leq \varepsilon n^{2/3}\log^{2/3}n$. The worst-case time per update operation is $O(\varepsilon^{-3})$. 
	\end{corollary}
	\begin{proof}
		If $\varepsilon < n^{-\frac23}$, then for each update, one can use the naive BFS algorithm to exactly compute $\ncc(G)$, which runs in time $O(n^2)=O(\varepsilon^{-3})$. If $\varepsilon \geq n^{-\frac23}$, we can apply Theorem~\ref{thm:random_ncc_espn} with parameters $p=\frac{1}{n^c}$, $\varepsilon' = \varepsilon n^{-1/3}\cdot\log^{2/3}n$, to obtain an $\overline{cc}$ for $\ncc(G)$ with an additive error $\varepsilon n^{2/3}\log^{2/3}n$. The corresponding dynamic algorithm has update time $O(\varepsilon^{-3}\cdot (1+\frac{\log(1/\varepsilon)}{\log n}))=O(\varepsilon^{-3})$.%
	\end{proof}
	
	\textbf{Remark:} We cannot expect to be able to get a constant-time algorithm for maintaining the number of connected components with an additive error of $1$ or a multiplicative error of $2$: Any such algorithm would be able to decide whether the graph is connected or not, contradicting the  $\Omega(\log n)$ lower bound for dynamically maintaining whether a graph is connected~\cite{PatrascuD06}.

		\end{document}